\documentclass[11pt]{article}

\usepackage{amsmath,amstext,amsthm,amssymb,amsfonts}

\usepackage{fullpage}
\usepackage{url}
\usepackage{bbm}
\usepackage{graphicx}
\usepackage{psfrag,subfig}
\usepackage{verbatim}
\usepackage{algorithmic}
\usepackage{algorithm}

\newtheorem{theorem}{Theorem}
\newtheorem{clm}{Claim}
\newtheorem{lemma}{Lemma}
\newtheorem{definition}{Definition}
\newtheorem{corollary}{Corollary}

\newtheorem{proposition}{Proposition}
\newtheorem{remark}{Remark}

\newcommand{\conv}{\textrm{conv}}
\newcommand{\E}{\mathbb{E}}
\newcommand{\diag}{\textrm{diag}}
\newcommand{\Supp}{\textrm{Supp}}

\begin{document}
\title{Approximating Nash Equilibria and Dense Subgraphs via an \\ Approximate Version of Carath\'{e}odory's Theorem}
\author{Siddharth Barman\thanks{California Institute of Technology. {\tt barman@caltech.edu}}}
\date{}

\maketitle

\begin{abstract}
We present algorithmic applications of an approximate version of Carath\'{e}odory's theorem. The theorem states that given a set of vectors $X$ in $\mathbb{R}^d$, for every vector in the convex hull of $X$ there exists an $\varepsilon$-close (under the $p$-norm distance, for $2\leq p < \infty$) vector that can be expressed as a convex combination of at most $b$ vectors of $X$, where the bound $b$ depends on $\varepsilon$ and the norm $p$ and is \emph{independent} of  the dimension $d$. This theorem can be derived by instantiating Maurey's lemma, early references to which can be found in the work of Pisier (1981) and Carl (1985). However, in this paper we present a self-contained proof of this result. 

Using this theorem we establish that in a bimatrix game with $ n \times n$ payoff matrices $A, B$, if the number of non-zero entries in any column of $A+B$ is at most $s$ then an $\varepsilon$-Nash equilibrium of the game can be computed in time $n^{O\left(\frac{\log s }{\varepsilon^2}\right)}$. This, in particular, gives us a polynomial-time approximation scheme for Nash equilibrium in games with fixed column sparsity $s$. Moreover, for arbitrary bimatrix games---since $s$ can be at most $n$---the running time of our algorithm matches the best-known upper bound, which was obtained by Lipton, Markakis, and Mehta (2003).

The approximate Carath\'{e}odory's theorem also leads to an \emph{additive} approximation algorithm  for the normalized densest $k$-subgraph problem. Given a graph with $n$ vertices and maximum degree $d$, the developed algorithm determines a subgraph with exactly $k$ vertices with normalized density within $\varepsilon$ (in the additive sense) of the optimal in time $n^{O\left( \frac{\log d}{\varepsilon^2}\right)}$. Additionally, we show that a similar approximation result can be achieved for the problem of finding a $k \times k$-bipartite subgraph of maximum normalized density. 






\end{abstract}


\section{Introduction}
Carath\'{e}odory's theorem is a fundamental dimensionality result in convex geometry. It states that any vector in the convex hull of a set $X$  in $\mathbb{R}^d$ can be expressed as a convex combination of at most $d+1$ vectors of $X$.\footnote{This bound of $d+1$ is tight.} This paper considers a natural approximate version of  Carath\'{e}odory's theorem where the goal is to seek convex combinations that are close enough to vectors in the convex hull. Specifically, this approximate version establishes that given a set of vectors $X$ in the $p$-unit ball\footnote{That is, $X$ is contained in the set $\{ v \in \mathbb{R}^d \mid \| v  \|_p \leq 1\}$. } with norm $p \in [2, \infty)$, for every vector $\mu$ in the convex hull of $X$ there exists an $\varepsilon$-close---under the $p$-norm distance---vector $\mu'$ that  can be expressed as a convex combination of $\frac{4 p }{\varepsilon^2}$ vectors of $X$.  A notable aspect of this result is that the number of vectors of $X$ that are required to express $\mu'$, i.e., $\frac{4 p}{\varepsilon^2}$, is independent of the underlying dimension $d$. This theorem can be derived by instantiating Maurey's lemma, early references to which can be found in the work of Pisier~\cite{pisier1981remarques} and Carl~\cite{carl1985inequalities}. However, in this paper we present a self-contained proof of this result, which we proceed to outline below. The author was made aware of the connection with Maurey's lemma after a preliminary version of this work had appeared. 


To establish the approximate version of Carath\'{e}odory's theorem we use the probabilistic method. Given a vector $\mu$ in the convex hull of a set $X \subset \mathbb{R}^d$, consider a convex combination of vectors of $X$ that generates $\mu$. The coefficients in this convex combination induce a probability distribution over $X$ and the mean of this distribution is $\mu$. The approach is to draw $b$ independent and identically distributed (i.i.d.)\ samples from this distribution and show that with positive probability the sample mean, with an appropriate number of samples, is close to $\mu$ under the $p$-norm distance, for $p \in [2, \infty)$. Therefore, the probabilistic method implies that these exists a vector close to $\mu$ that can be expressed as a convex combination of at most $b$ vectors, where $b$ is the number of samples we drew. 

Note that in this context applying the probabilistic method is a natural idea, but a direct application of this method will not work. Specifically, a dimension-free result is unlikely if we first try to prove that the $i$th component of the sample mean vector is close to the $i$th component of $\mu$, for every $i \in [d]$; since this would entail a union bound over the number of components $d$. Bypassing such a component-wise analysis requires the use of atypical    ideas. We are able to accomplish this task and, in particular, bound (in expectation) the $p$-norm distance between $\mu$ and the sample mean vector via an interesting application of Khintchine inequality (see Theorem~\ref{thm:Khintchine-v}).



Given the significance of Carath\'{e}odory's theorem, this approximate version is interesting in its own right. The key contribution of the paper is to substantiate the algorithmic relevance of this approximate version by developing new algorithmic applications. Our applications include additive approximation algorithms for (i) Nash equilibria in two-player games, and (ii) the densest subgraph problem. These algorithmic results are outlined below. 

\subsection*{Algorithmic Applications}

\paragraph{Approximate Nash Equilibria.} Nash equilibria are central constructs in game theory that are used to model likely outcomes of strategic interactions between self-interested entities, like human players. They denote distributions over actions of players under which no player can benefit, in expectation, by unilateral deviation. These solution concepts are arguably the most well-studied notions of rationality and questions about their computational complexity lie at the core of algorithmic game theory.  In recent years, hardness results have been established for Nash equilibrium, even in two-player games~\cite{CD,DGP}. But, the question whether an \emph{approximate} Nash equilibrium can be computed in polynomial time still remains open. Throughout this paper we will consider the standard additive notion of approximate Nash equilibria that are defined as follows: a pair  distributions, one for each player, is said to be an $\varepsilon$-Nash equilibrium if any unilateral deviation increases utility by at most $\varepsilon$, in expectation.

We apply the approximate version of Carath\'{e}odory's theorem to address this central open question. Specifically, we prove that in a bimatrix game with $n \times n$ payoff matrices $A, B$, i.e., a two-player game with $n$ actions for each player, if the number of non-zero entries in any column of $A+B$ is at most $s$ then an $\varepsilon$-Nash equilibrium of the game can be computed in time $n^{O\left(\frac{\log s}{\varepsilon^2}\right)}$. Our result, in particular, shows that games with fixed column sparsity $s$ admit a polynomial-time approximation scheme (PTAS) for Nash equilibrium. Along the lines of zero-sum games (which model strict competition), games with fixed column sparsity capture settings in which, except for particular action profiles, the gains and losses of the two player balance out. In other words, such games are a natural generalization of zero-sum games; recall that zero-sum games admit efficient computation of Nash equilibrium (see, e.g.,~\cite{AGTbook}). 

It is also worth pointing out that for an arbitrary bimatrix game the running time of our algorithm is $n^{O\left(\frac{\log n}{\varepsilon^2}\right)}$, since $s$ is at most $n$. Given that the best-known algorithm for computing $\varepsilon$-Nash equilibrium also runs in time $n^{O\left(\frac{\log n}{\varepsilon^2}\right)}$~\cite{LMM}, for general games the time complexity of our algorithm matches the best-known upper bound. Overall, this result provides a parameterized understanding of the complexity of computing approximate Nash equilibrium in terms of a very natural measure, the column sparsity $s$ of the matrix $A+B$. 

Our framework can address other notions of sparsity as well. Specifically, if there exist \emph{constants} $\alpha, \beta \in \mathbb{R}_+$ and $\gamma \in \mathbb{R}$ such that the matrix $\alpha A + \beta B + \gamma \mathbbm{1}_{n \times n}$ has column or row sparsity $s$, then our algorithm can be directly adopted to find an $\varepsilon$-Nash equilibrium of the game $(A,B)$ in time $n^{O\left(\frac{\log s}{\varepsilon^2}\right)}$; here, $\mathbbm{1}_{n \times n}$ is the all-ones $n \times n$ matrix.\footnote{Note that given matrices $A$ and $B$, parameters $\alpha$, $\beta$, and $\gamma$ can efficiently computed.} Additionally, the same running-time bound can be achieved for approximating Nash equilibrium in games wherein \emph{both} matrices $A$ and $B$ have column or row sparsity $s$. Note that this case is not subsumed by the previous result; in particular, if the columns of matrix $A$ and the rows of matrix $B$ are sparse, then it is not necessary that $A+B$ has low column or row sparsity.  


We also refine the following result of Daskalakis and Papadimitriou~\cite{smallprob}: They develop a PTAS for bimatrix games that admit an equilibrium with small, specifically $O\left(\frac{1}{n}\right)$, probability values. This result is somewhat surprising, since such small-probability equilibria have large, $\Omega(n)$, support, and hence are not amenable to, say, exhaustive search. We show that if a game has an equilibrium with probability values $O\left( \frac{1}{m} \right)$, for $m \in [n]$, then an approximate equilibrium can be computed in time $n^t$, where $t= O\left(\frac{\log (s/m)}{\varepsilon^2}\right)$. Since $s\leq n$, we get the result of~\cite{smallprob} as a special case.


\paragraph{Densest Subgraph.} In the normalized densest $k$-subgraph problem (\rm{NDkS}) we are given a simple graph and the objective is to find a size-$k$ subgraph (i.e., a subgraph containing exactly $k$ vertices) of maximum density; here, density is normalized to be at most one, i.e., for a subgraph with $k$ vertices, it is defined to be the number of edges in the subgraph divided by $k^2$. \rm{NDkS} is simply a normalized version of of the standard densest $k$-subgraph problem (see, e.g.,~\cite{alon2011inapproximability} and references therein) wherein the goal is to find a  subgraph with $k$ vertices with the maximum possible number of edges in it. The densest $k$-subgraph problem (\rm{DkS}) is computationally hard and it is shown in~\cite{alon2011inapproximability} that a constant-factor approximation for \rm{DkS} is unlikely. This result implies that \rm{NDkS} is hard to approximate (multiplicatively) within a constant factor as well.

In this paper we focus on an additive approximation for \rm{NDkS}. In particular, our objective is to compute a size-$k$ subgraph whose density is close (in the additive sense) to the optimal. The paper also presents additive approximations for the densest $k$-bipartite subgraph (\rm{DkBS}) problem. \rm{DkBS} is a natural variant of \rm{NDkS} and the goal in this problem is to find size-$k$ vertex subsets of maximum density. In the bipartite case, density of vertex subsets $S$ and $T$ is defined to be the number of edges between the two subsets divided by $|S||T|$. 

Hardness of additively approximating \rm{DkBS} was studied by Hazan and Krauthgamer~\cite{hazan2011hard}. Specifically, the reduction in~\cite{hazan2011hard} rules out an additive PTAS for \rm{DkBS}, under complexity theoretic assumptions.\footnote{They reduce the problem of determining a \emph{planted clique} to that of computing an $\varepsilon$-additive approximation for \rm{DkBS}, with a sufficiently small but constant $\varepsilon$.} In terms of upper bound, the result of Alon et al.~\cite{alonvempala} presents an algorithm for this problem that runs in time exponential in the rank of the adjacency matrix. 

This paper develops the following complementary upper bounds: given a graph with $n$ vertices and maximum degree $d$, an $\varepsilon$-additive approximation for \rm{NDkS} can be computed in time $n^{O\left(\frac{\log d}{\varepsilon^2}\right)}$. This paper also presents an algorithm with the same time complexity for additively approximating \rm{DkBS}.

\subsection{Related Work}

\paragraph{Approximate Version of Carath\'{e}odory's Theorem.} In this paper we provide a self-contained proof of the approximate version of Carath\'{e}odory's theorem, employing the Khintchine inequality (see Theorem~\ref{thm:Khintchine-v}), and use the theorem to develop new approximation algorithms. As mentioned earlier, the approximate version of Carath\'{e}odory's theorem can also be obtained by instantiating Maurey's lemma, which, in particular, appears in the analysis and operator theory literatures; see, e.g.,~\cite{pisier1981remarques, carl1985inequalities, bourgain1989duality}. 


\paragraph{Approximate Nash Equilibria.} The computation of equilibria is an active area of research. Nash equilibria is known to be computationally hard~\cite{CD,DGP}, and in light of these findings, a considerable effort has been directed towards understanding the complexity of \emph{approximate} Nash equilibrium. Results in this direction include both upper bounds~\cite{LMM,KPS,daskalakis2006note, KT, DMPprogress, kontogiannis2007efficient, feder2007approximating, bosse2007new, tsaknakis2007optimization, tsaknakis2010practical, alonvempala, alon2014cover} and lower bounds~\cite{hazan2011hard,daskalakis2013complexity, bravermanapproximating}.  In particular, it is known that for a general bimatrix game an approximate Nash equilibrium can be computed in quasi-polynomial time~\cite{LMM}. Polynomial time algorithms have been developed for computing approximate Nash equilibria for fixed  values of the approximation factor $\varepsilon$; the best-known result of this type shows that a $0.3393$-approximate Nash equilibrium can be computed in polynomial time~\cite{tsaknakis2007optimization}. In addition, several interesting classes of games have been identified that admit a PTAS~\cite{KT, smallprob, tsaknakis2010practical, alonvempala, alon2014cover}. For example, the result of Alon et al.~\cite{alonvempala} provides a PTAS for games in which the sum of the payoff matrices, i.e., $A+B$, has logarithmic rank. Our result is incomparable to such rank based results, since a sparse matrix can have high rank and a low-rank matrix can have high sparsity.



Chen et al.~\cite{chen2006sparse} considered sparsity in the context of games and showed that computing an exact Nash equilibrium is hard even if \emph{both} the payoff matrices have a fixed number of non-zero entries in every row \emph{and} column. It was observed in~\cite{smallprob} that such games admit a trivial PTAS.\footnote{In particular, the product of uniform distributions over players' actions corresponds to an approximate Nash equilibrium in such games.} Note that we study a strictly larger class of games and provide a PTAS for games in which the row \emph{or} column sparsity of $A+B$ is fixed.

%
\paragraph{Densest Subgraph.} The best-known (multiplicative) approximation ratio for the densest $k$-subgraph problem is $n^{(1/4+ o(1))}$~\cite{bhaskara2010detecting}. But unlike this result, our work addresses additive approximations with normalized density as the maximization objective. In parituclar, we approximate \rm{NDkS} by approximately solving a quadratic program, which is similar to the quadratic program used in the Motzkin-Straus theorem~\cite{motzkin_maxima_1965}. 

In addition, our approximation algorithm for \rm{DkBS} is based on solving a bilinear program that was formulated by Alon et al.~\cite{alonvempala}. This bilinear program was used in~\cite{alonvempala} to develop an additive PTAS for \rm{DkBS} in particular classes of graphs, including ones with low-rank adjacency matrices. This paper supplements prior work by developing an approximation algorithm whose running time is parametrized by the maximum degree of the given graph, and not by the rank of its adjacency matrix.
 

\subsection{Techniques}
\label{sect:tech}

\paragraph{Approximate Nash Equilibria.} Our algorithm for computing an approximate Nash equilibrium relies on finding a near-optimal solution of a bilinear program (BP).  The BP we consider was formulated by Mangasarian and Stone~\cite{MS} and its optimal (near-optimal) solutions correspond to exact (approximate) Nash equilibria of the given game. 
Below we provide a sketch of our algorithm that determines a near-optimal solution of this BP.

The variables of the BP, $x$ and $y$, correspond to probability distributions that are mixed strategies of the players and its objective is to maximize $x^T C y$, where $C$ is the sum of the payoff matrices of the game.\footnote{We ignore the linear part of the objective for ease of presentation, see Section~\ref{sect:nash} for details.}  Suppose we knew the vector $u:=C \hat{y}$, for some Nash equilibrium $(\hat{x}, \hat{y})$. Then, a Nash equilibrium can be efficiently computed by solving a linear program (with variables $x$ and $y$) that is obtained by modifying the BP as follows: replace $x^T C y $ by $x^Tu$ as the objective and include the constraint $Cy = u$. Section~\ref{sect:nash} shows that this idea can be used to find an approximate Nash equilibrium, even if $u$ is not exactly equal to $C\hat{y}$ but close to it.  That is, to find an approximate Nash equilibrium it suffices to have a vector $u$ for which $\| C\hat{y} - u \|_p$ is small.

To apply the approximate version of Carath\'{e}odory's theorem we observe that $C\hat{y}$ is a vector in the convex hull of the \emph{columns} of $C$. Also, note that in the context of (additive) approximate Nash equilibria the payoff matrices are normalized, hence  the absolute value of any entry of matrix $C$ is no more than, say, $2$. This entry-wise normalization implies that if no column of matrix $C$ has more than $s$ non-zero entries, then the $\log s$ norm of the columns is a fixed constant: $\| C^i \|_p \leq (s \cdot 2^p )^{1/p} = 2 \cdot 2^{\frac{\log s}{p}} \leq 4 $, where $C^i$ is the $i$th column of $C$ and norm $p = \log s$. This is a simple but critical observation, since it implies that, modulo a small scaling factor, the columns of an $C$ lie in the $\log s$-unit ball.  At this point we can apply the approximate version of Carath\'{e}odory's theorem to guarantee that close to $C \hat{y}$ there exists a vector $u$ that can be expressed as a convex combination of about $p=\log s$ columns of $C$. We show in Section~\ref{sect:nash} that exhaustively searching for $u$ takes $n^{O(\log s)}$ time, where $n$ is the number of columns of $C$. Thus we can find a vector close to $C \hat{y}$ and hence determine a near-optimal solution of the bilinear program. This way we get an approximate Nash equilibrium and the running time of the algorithm is dominated by the exhaustive search.

Overall, this template for approximating Nash equilibria in sparse games is made possible by the approximate version of Carath\'{e}odory's theorem. It is notable that our algorithmic framework employs arbitrary norms $p \in [2, \infty)$, and in this sense it goes beyond standard \emph{$\varepsilon$-net}-based results that typically use norms $1$, $2$, or $\infty$.  
\paragraph{Densest Subgraph.} The algorithmic approach outlined above applies to any quadratic or bilinear program in which the objective matrix is column (or row) sparse and the feasible region is contained in the simplex. We use this observation to develop an additive approximations for \rm{NDkS} and \rm{DkBS}.

Specifically, we formulate a quadratic program, near-optimal solutions of which correspond to approximate-solutions of \rm{NDkS}. The column sparsity of the objective matrix in the quadratic program is equal to the maximum degree of the underlying graph plus one. Hence, using the above mentioned observation, we obtain the approximation result for \rm{NDkS}. The same template applies to \rm{DkBS}; for this problem we employ a bilinear program from~\cite{alonvempala}. 



\subsection{Organization}
We begin by setting up notation in Section~\ref{sect:notation}. Then, in Section~\ref{sect:caratheodory} we present the approximate version of Carath\'{e}odory's theorem. Algorithmic applications of the theorem are developed in Sections~\ref{sect:nash} and~\ref{sect:dense}. In Section~\ref{sect:ext} we consider convex hulls of matrices and also detail approximate versions of the colorful Carath\'{e}odory theorem and Tverberg's theorem. Finally, Section~\ref{sect:lb} presents a lower bound proving showing that, in general, $\varepsilon$-close (under the $p$-norm distance with $p \in [2, \infty)$) vectors cannot be expressed as a convex combination of less than $ \frac{1}{4 \ \varepsilon^{p/(p-1)}} $ vectors of the given set.

\section{Notation}
\label{sect:notation}
Write $\| x \|_p$ to denote the $p$-norm of a vector $x\in \mathbb{R}^d$. The Euclidean norm is denoted by $\| x \|$, i.e., we drop the subscript from $\| x \|_2$. The number of non-zero components of a vector $x$ is specified via the $\ell_0$ ``norm'': $\| x \|_0 := | \{ i \mid x_i \neq 0\} |$. Let $\Delta^n$ be the set of probability distributions over the set $[n]$. For $x \in \Delta^n$, we define $\Supp(x) := \{ i  \mid  x_i \neq 0\}$. Similarly, for a vector $v \in \mathbb{R}^n$ write $\Supp(v)$ to denote the set $\{ i  \mid v_i \neq 0 \}$.

Given a set $X=\{ x_1, x_2, \ldots, x_n \} \subset \mathbb{R}^d$, we use the standard abbreviation $\conv(X)$ for the convex hull of $X$. A vector $y \in \conv(X)$ is said to be \emph{$k$ uniform} with respect to $X$ if there exists a size $k$ multiset $S$ of $[n]$ such
that $ y = \frac{1}{k} \sum_{i \in S} x_i $. In particular, if vector $y$ is $k$ uniform with respect to $X$ then $y$
can be expressed as a convex combination of at most $k$ vectors from $X$. Throughout, the set $X$ will be clear
from context so we will simply say that a vector is $k$ uniform and not explicitly mention the fact that uniformity is with respect to $X$.

\section{Approximate Version of Carath\'{e}odory's Theorem}
\label{sect:caratheodory}

A key technical ingredient in our  proof is  Khintchine inequality (see, e.g.,~\cite{garling} and~\cite{lindenstrauss1973classical}). The following form of the inequality is derived from a result stated in~\cite{So}.

\begin{theorem}[Khintchine Inequality]
\label{thm:Khintchine-v}
Let $r_1, r_2, \ldots, r_m$ be a sequence of i.i.d.\ Rademacher $\pm 1$ random variables, i.e., $\Pr( r_i = \pm 1) = \frac{1}{2} $
for all $i \in [m]$. In addition, let $u_1, u_2, \ldots, u_m \in \mathbb{R}^d$ be a deterministic sequence of vectors. Then, for $2 \leq p <\infty$

\begin{align}
\E  \left\| \sum_{i=1}^m r_i u_i \right\|_p  \leq  \sqrt{p} \ \left( \sum_{i=1}^m \|u_i\|_p^2 \right)^{\frac{1}{2}}.
\end{align}
\end{theorem}
\begin{proof}
Given vector $v \in \mathbb{R}^d$, write $\diag(v)$ to denote the $d \times d$ diagonal matrix whose diagonal is equal to $v$. Note that if matrix $Q = \diag(v)$ then $\| Q \|_{S_p} = \| v \|_p$, where $\| Q \|_{S_p}$ denotes the Schatten $p$-norm of $Q$, i.e., $\| Q \|_{S_p} = \| \sigma(Q) \|_p$, where $\sigma(Q)$ is the vector of singular values of $Q$.  In addition, say we construct diagonal matrices for a sequence of vectors $u_1, u_2, \ldots, u_m \in \mathbb{R}^d$, i.e., set $Q_i = \diag(u_i)$ for all $i \in [m]$, then for any sequence of scalars $\xi_1, \xi_2, \ldots, \xi_m \in \mathbb{R}$ we have $\sum_{i=1}^m \xi_i Q_i  = \diag\left(  \sum_{i=1}^m \xi_i u_i  \right)$.

In order to prove the theorem statement for vectors $u_1, u_2, \ldots, u_m \in \mathbb{R}^d$, we simply use Theorem 2 of~\cite{So}. In particular,  setting  $Q_i = \diag(u_i)$ for all $i \in [m]$ in Theorem 2 of~\cite{So} and using the above stated observations we get:

\begin{align}
\E  \left\| \sum_{i=1}^m r_i u_i \right\|_p^p  \leq  p^{p/2} \ \left( \sum_{i=1}^m \|u_i\|_p^2 \right)^{\frac{p}{2}}. \label{ineq:So}
\end{align}

For $p\geq 2$ the $p$th root is a concave function; hence Jensen's inequality, applied to (\ref{ineq:So}), gives us the desired result.

\end{proof}


We are ready to prove the main result of this section. Note that in the following theorem the scaling term $\gamma$ is defined with respect to the $p$ norm.

\vspace*{12pt}

\begin{theorem}
\label{thm:caratheodory-p}
Given a set of vectors $X= \{ x_1, x_2,\ldots, x_n \} \subset \mathbb{R}^d$ and $\varepsilon >0$. For every $\mu \in \conv(X)$ and $2 \leq p < \infty $ there exists an  $\frac{{4p} \gamma^2}{\varepsilon^2}$ uniform vector $\mu' \in \conv(X)$ such that $\| \mu - \mu' \|_p \leq \varepsilon$. Here, $\gamma := \max_{x \in X} \| x \|_p$.
\end{theorem}

\begin{proof}
Express  $\mu \in \conv(X)$ as a convex combination of $x_i$s: $ \mu =  \sum_{i=1}^n \alpha_i x_i$ where $\alpha_i \geq 0 $, for all
$i \in [n]$, and $\sum_{i=1}^n \alpha_i = 1$. Note that $\alpha = (\alpha_1, \alpha_2, \ldots, \alpha_n)$ corresponds to a
probability distribution over vectors $ x_1, x_2, \ldots, x_n$. That is, under probability distribution $\alpha$ vector $x_i$ is drawn
with probability $\alpha_i$. The vector $\mu$ is the mean of this distribution. Specifically, the $j$th component of $\mu$ is the
expected value of the random variable that takes value $x_{i,j}$ with probability $\alpha_i$, here $x_{i,j}$ is the $j$th component of
vector $x_i$. We succinctly express these component-wise equalities as follows:
\begin{align}
\E_{v \sim \alpha} [v] = \mu.
\end{align}

Let $v_1, v_2,\dots, v_m$ be $m$ i.i.d.\ draws from $\alpha$. The sample mean vector is defined to be $\frac{1}{m} \sum_{i=1}^m v_i$. Below we specify function $g: X^m \rightarrow \mathbb{R}$ to quantify the $p$-norm distance between the sample mean vector and the $\mu$.

\begin{align}
g(v_1,v_2,\ldots, v_m) := \left\| \frac{1}{m} \sum_{i=1}^m v_i - \mu \right\|_p.
\end{align}

The key technical part of the remainder of the proof is to show that \begin{align}
\E[g] \leq \frac{2 \sqrt{p} \  \gamma}{\sqrt{m}}. \end{align}

For $m = \frac{4p \gamma^2}{\varepsilon^2}$ this inequality reduces to $\E[g] \leq \varepsilon$. Therefore, when the number of samples $m = \frac{4p \gamma^2}{\varepsilon^2}$ we have $\Pr(g \leq \varepsilon) > 0$, i.e., with positive probability the sample mean vector is $\varepsilon$ close to $\mu$ in the $p$-norm. Overall, the stated claim is implied by the probabilistic method.

Recall that in expectation the sampled mean is equal to $\mu$, i.e., $\E_{v_1',\ldots, v_m' \sim \alpha} \frac{1}{m} \sum_{i=1}^m v_i' = \mu$. Hence,
we have

\begin{align}
\E[g] & = \E_{v_1,\ldots, v_m} \left\| \frac{1}{m} \sum_{i=1}^m v_i  \ - \ \mu \right\|_p \\
& = \E_{v_1,\ldots, v_m} \left\| \frac{1}{m} \sum_{i=1}^m v_i  \ - \  \E_{v_1',\ldots, v_m'} \frac{1}{m} \sum_{i=1}^m v_i' \right\|_p \\
& = \E_{v_1,\ldots, v_m} \left\| \E_{v_1',\ldots, v_m'} \left( \frac{1}{m} \sum_{i=1}^m v_i  \ - \   \frac{1}{m} \sum_{i=1}^m v_i' \right) \right\|_p.\label{ineq:break}
\end{align}

Note that $\| \cdot \|_p$ is convex for $p\geq 1$. Therefore, Jensen's inequality gives us:
\begin{align}
\E_{v_1,\ldots, v_m} \left\| \E_{v_1',\ldots, v_m'} \left( \frac{1}{m} \sum_{i=1}^m v_i  \ - \   \frac{1}{m} \sum_{i=1}^m v_i' \right) \right\|_p & \leq \E_{v_1,\ldots, v_m} \E_{v_1',\ldots, v_m'} \left\|  \left( \frac{1}{m} \sum_{i=1}^m v_i  \ - \   \frac{1}{m} \sum_{i=1}^m v_i' \right) \right\|_p \\
& = \frac{1}{m} \E_{\substack{v_1,\ldots, v_m \\ v_1',\ldots, v_m'}} \left\| \left( \sum_{i=1}^m \left( v_i  \ - \ v_i' \right) \right) \right\|_p. \label{ineq:sub}
\end{align}

Let $r_1, r_2, \ldots, r_m$ be a sequence of i.i.d.\ Rademacher $\pm 1$ random variables, i.e., $\Pr( r_i = \pm 1) = \frac{1}{2} $
for all $i \in [m]$. Since, for all $i \in [m]$, $v_i$ and $v_i'$ are i.i.d.\ copies we can write
\begin{align}
\frac{1}{m} \E_{v_i,v_i'} \left\| \left( \sum_{i=1}^m \left( v_i  - v_i' \right) \right) \right\|_p & = \frac{1}{m} \E_{v_i,v_i', r_i} \left\| \left( \sum_{i=1}^m r_i \left( v_i  - v_i' \right) \right) \right\|_p \nonumber \\
& \leq \frac{1}{m} \E_{v_i,v_i', r_i} \left[ \left\| \sum_{i=1}^m r_i v_i \right\|_p + \left\| \sum_{i=1}^m r_i v_i' \right\|_p \right] \quad \textrm{(Triangle inequality)} \nonumber \\
& = \frac{1}{m} \E_{r_i} \left[ \ \E_{v_i, v_i'} \left( \left\| \sum_{i=1}^m r_i v_i \right\|_p + \left\| \sum_{i=1}^m r_i v_i' \right\|_p \ \ \middle| \  r_1,..,r_m \right) \  \right] \ \ \textrm{(Tower property)} \nonumber \\
& = \frac{1}{m} \E_{r_i} \left[ \ \E_{v_i} \left( \left\| \sum_{i=1}^m r_i v_i \right\|_p \ \ \middle| \  r_1,..,r_m \right) \ + \  \E_{v_i'} \left( \left\| \sum_{i=1}^m r_i v_i' \right\|_p \ \ \middle|  \  r_1,..,r_m \right) \  \right] \nonumber \\
& = \frac{1}{m} \E_{r_i} \left[ \ 2 \ \E_{v_i} \left( \left\| \sum_{i=1}^m r_i v_i \right\|_p \ \ \middle| \  r_1,..,r_m \right) \  \right] \nonumber \\
& = 2 \ \E_{v_i, r_i} \left\| \sum_{i=1}^m r_i \ \frac{v_i}{m} \right\|_p. \label{ineq:rad}
\end{align}

The penultimate equality follows from the following ($v_i$s and $v_i'$s are i.i.d.\ copies)
\begin{align}
\E_{v_i} \left( \left\| \sum_{i=1}^m r_i v_i \right\|_p \ \ \middle| \  r_1,..,r_m \right) & =  \E_{v_i'} \left( \left\| \sum_{i=1}^m r_i v_i' \right\|_p \ \ \middle|  \  r_1,..,r_m \right).
\end{align}

Overall, inequalities (\ref{ineq:break}), (\ref{ineq:sub}), and (\ref{ineq:rad}) imply
\begin{align}
\E[g] \leq 2 \ \E_{v_i, r_i} \left\| \sum_{i=1}^m r_i \ \frac{v_i}{m} \right\|_p \label{ineq:exprad},
\end{align}
where $r_1, r_2, \ldots, r_m$ is a sequence of i.i.d.\ Rademacher $\pm 1$ random variables.

At this point we can apply Theorem~\ref{thm:Khintchine-v} (Khintchine inequality) with $u_i = \frac{v_i}{m}$ to obtain
\begin{align}
\E_{v_i, r_i} \left\| \sum_{i=1}^m r_i \ \frac{v_i}{m} \right\|_p & = \E_{v_i} \left[ \ \E_{r_i} \left[ \left\| \sum_{i=1}^m r_i \ \frac{v_i}{m} \right\|_p \ \ \middle| \ v_1,..,v_m \right] \  \right] \\
& \leq \E_{v_i} \left[   \sqrt{p} \ \left( \sum_{i=1}^m \left\| \frac{v_i}{m} \right\|_p^2 \right)^{1/2} \right] \\
& \leq \E_{v_i} \left[  \sqrt{p} \ \left( \sum_{i=1}^m \frac{\gamma^2}{m^2} \right)^{1/2} \right] \label{ineq:normuse} \\
& =  \sqrt{p} \ \frac{\gamma}{\sqrt{m}}. \label{ineq:finalexp}
\end{align}

Inequality (\ref{ineq:normuse}) uses the fact that random vectors $v_i$ are supported over $X$, so $\|v_i\|^2_p \leq \gamma^2$.

Using (\ref{ineq:exprad}) and (\ref{ineq:finalexp}) we get
\begin{align}
\E[g] \leq \frac{2 \sqrt{p} \ \gamma}{\sqrt{m}}.
\end{align}

This completes the proof. 



\end{proof}

We end this section by stating an $\infty$-norm variant of our result. This theorem follows directly from Hoeffding's inequality. Note that in the following theorem the scaling of vectors in $X$ is with respect to the $\infty$-norm.

\begin{theorem}
\label{thm:caratheodory-infty}
Given a set of vectors $X= \{ x_1, x_2,\ldots, x_n \} \subset \mathbb{R}^d$, with $ \max_{x \in X} \| x \|_\infty \leq 1$, and $\varepsilon >0$. For every $\mu \in \conv(X)$ there exists an $O\left( \frac{\log n}{\varepsilon^2}\right)$ uniform vector $\mu' \in \conv(X)$ such that $\| \mu - \mu' \|_\infty \leq \varepsilon$.
\end{theorem}

\begin{proof}
Apply Hoeffding's inequality component wise and take union bound.
\end{proof}

\section{Computing Approximate Nash Equilibrium}
\label{sect:nash}

\textbf{Bimatrix Games.} Bimatrix games are two player games in normal
form.  Such games are specified by a pair of $n \times n$ matrices $(A, B)$, which are termed the payoff matrices for the
players. The first player, also called the row player, has payoff
matrix $A$, and the second player, or the column player, has payoff
matrix $B$. The strategy set for each player is
$[n]=\{1,2,\ldots,n\}$, and, if the row player plays strategy $i$ and
column player plays strategy $j$, then the payoffs of the two players
are $A_{ij}$ and $B_{ij}$ respectively. The payoffs of the players are normalized between $-1$ and $1$, so $A_{ij}, B_{ij} \in [-1,1]$ for all $i,j \in [n]$.

Recall that $\Delta^n$ is the set of probability distributions over the set of pure strategies $[n]$. We use $e_i \in \mathbb{R}^n$ to denote the vector with $1$ in the $i$th coordinate and $0$'s elsewhere. The players can randomize over their strategies by selecting any probability distribution in $\Delta^n$, called a mixed strategy. When the row and column players play mixed strategies $x$ and $y$ respectively, the expected payoff of the row player is $x^T A y $ and the expected payoff of the column player is $x^T By$.

\begin{definition}[Nash Equilibrium]
A mixed strategy pair $(x,y)$,  $ x,y \in \Delta^n$, is said to be a Nash equilibrium if and only if:
\begin{align}
x^T A y &  \geq e_i^T A y \qquad \forall \ i \in [n] \ \ \textrm{ and }\\
x^T B y & \geq x^T B e_j \qquad \forall \ j \in [n].
\end{align}
\end{definition}

By definition, if $(x,y)$ is a Nash equilibrium neither the row player nor the column player can benefit, in expectation, by unilaterally deviating to some other strategy. We say that a mixed strategy pair is an $\varepsilon$-Nash equilibrium is no player can benefit more than $\varepsilon$, in expectation, by unilateral deviation. Formally,

\begin{definition}[$\varepsilon$-Nash Equilibrium]
A mixed strategy pair $(x,y)$,  $ x,y  \in \Delta^n$, is said to be an  $\varepsilon$-Nash equilibrium if and only if:
\begin{align}
x^T A y &  \geq e_i^T A y - \varepsilon \qquad \forall \ i \in [n] \ \ \textrm{ and }\\
x^T B y & \geq x^T B e_j - \varepsilon  \qquad \forall \ j \in [n].
\end{align}
\end{definition}

Throughout, we will write $C$ to denote the sum of payoff matrices, $C:=A+B$. We will denote the $i$th column of $C$ by $C^i$, for $i \in [n]$. Note that $\| C^i \|_0$ is equal to the number of non-zero entries in the $i$th column of $C$.

In the following definition we ensure that the sparsity parameter $s$ is at least $4$ for ease of presentation. In particular, the running time of our algorithm depends on the $\log$ of the number of non-zero entries in the columns of $C$, i.e., $\log$ of the sparsity of the columns of matrix $C$. Setting $s \geq 4$ gives us $\log s \geq 2$. This allows us to state a single running-time bound, which holds even for corner cases wherein the column sparsity is, say, zero (and hence $\log (\max_{i} \| C^i \|_0)$ is undefined).

\begin{definition}[$s$-Sparse Games]
The sparsity of a game $(A,B)$ is defined to be \[ s := \max\{ \max_{i} \| C^i \|_0,  4 \}, \]where matrix $C=A+B$.
\end{definition}


The quantitative connection between the sparsity of a game and the time it takes to compute an $\varepsilon$-Nash equilibrium is stated below.
\begin{theorem}
\label{thm:nash}
Let  $A, B \in [-1,1]^{n \times n}$ be the payoff matrices of an $s$-sparse bimatrix game. Then, an $\varepsilon$-Nash equilibrium of $(A,B)$ can be computed in time
\begin{align*}
n^{ O \left( \frac{\log s}{\varepsilon^2} \right) }.
\end{align*}
\end{theorem}

\vspace*{10pt}

Our algorithm for computing $\varepsilon$-Nash equilibrium relies on the following bilinear program, which was formulated by Mangasarian and Stone~\cite{MS}. As formally specified in Lemma~\ref{lem:nash} below, approximate solutions of this bilinear program correspond to approximate Nash equilibria.

\begin{align}
\max_{x,y, \pi_1, \pi_2} & \ \ \ x^TCy - \pi_1 - \pi_2  \nonumber \\
\textrm{subject to } \ \ \ & x^T B  \leq \mathbbm{1}^T \pi_2   \nonumber \\
&  A y  \leq \mathbbm{1} \pi_1  \nonumber \\
& x , y \in \Delta^n \nonumber \\
& \pi_1, \pi_2 \in [-1,1]. \tag{BP}
\end{align}

Here $\mathbbm{1}$ denotes the all-ones vector.  Using the definition of Nash equilibrium one can show that the optimal solutions of this bilinear program correspond to Nash equilibria of the game $(A,B)$. Formally, we have

\begin{theorem}[Equivalence Theorem~\cite{MS}]
\label{thm:ms}
Mixed strategy pair $(\hat{x}, \hat{y})$ is a Nash equilibrium of the game $(A,B)$ if and only if $\hat{x}$, $\hat{y}$, $\hat{\pi}_1$, and $\hat{\pi}_2$ form an optimal solution of the bilinear program (BP), for some scalars $\hat{\pi}_1$ and $\hat{\pi}_2$. In addition, the optimal value achieved by (BP) is equal to zero and the payoffs of the row and column player at this equilibrium are $\hat{\pi}_1$ and $\hat{\pi}_2$ respectively.
\end{theorem}

A relevant observation is that that an approximate solution of the bilinear program corresponds to an $\varepsilon$-Nash equilibrium.

\begin{lemma}
\label{lem:nash}
Let $x, y \in \Delta^n$ along with scalars $\pi_1$ and $\pi_2$ form a feasible solution of (BP) that achieves an objective function value more than $-\varepsilon$, i.e., $x^TC y \geq \pi_1 + \pi_2 - \varepsilon$. Then, $(x,y)$ is an $\varepsilon$-Nash equilibrium of the game $(A,B)$.
\end{lemma}
\begin{proof}
The feasibility of $x,y$ implies that $ \max_j x^T B e_j \leq \pi_2$ and $\max_i e_i^T A y \leq \pi_1$. Since the objective function value achieved by $x,y$ is at least $-\varepsilon$ we have $x^T A y + x^T B y -  \pi_1 - \pi_2 \geq - \varepsilon$. But, $x^T A y $ is at most $\pi_1$ and $x^T B y $ is at most $\pi_2$.
So, the following inequalities must hold
\begin{align}
x^T A y & \geq \pi_1 - \varepsilon \qquad \textrm{and} \\
x^T B y & \geq \pi_2 - \varepsilon.
\end{align}

Overall, we get $x^T A y \geq \max_i e_i^T A y  - \varepsilon$ and $x^T B y  \geq \max_j x^T B e_j  - \varepsilon$. Hence $(x,y)$ satisfies the definition of an $\varepsilon$-Nash equilibrium.
\end{proof}

The proposed algorithm (see Algorithm~\ref{alg:nash}) solves the following $p$-norm minimization problem, with $p \geq 2$. The program CP($u$) is parametrized by vector $u \in \mathbb{R}^n$ and it can be solved in polynomial time.\footnote{Note that for fixed  $u$, CP($u$) is a convex program. Specifically, given $u \in \mathbb{R}^n$ and matrix $C \in \mathbb{R}^{n \times n}$, for $p \geq 1$, the function $f(x) := \| Cx - u \|_{p}$ is convex.}

\begin{align}
\min_{x,y,\pi_1,\pi_2} \ \ & \ \ \ \| Cy - u \|_p   \tag{\bf CP($u$)} \\
\textrm{subject to } \ \ \ & x^T u \geq \pi_1 + \pi_2 - \varepsilon/2 \nonumber \\
& Ay   \leq \mathbbm{1} \pi_1   \nonumber \\
& x^T B  \leq \mathbbm{1}^T \pi_2   \nonumber \\
& x,y \in \Delta^n \nonumber  \\
& \pi_1, \pi_2 \in [-1,1]. \nonumber
\end{align}

\begin{algorithm}{Given payoff matrices $A,B \in [-1,1]^{n \times n}$ and $\varepsilon >0$; Return: $\varepsilon$-Nash equilibrium of $(A,B)$}
\caption{Algorithm for computing $\varepsilon$-Nash equilibrium in $s$-sparse games}
  \begin{algorithmic}[1]
   \label{alg:nash}
   \STATE Write $s$ to denote the sparsity of the game $(A,B)$ and let $p = \log s $. \label{step:assign}\\
   \COMMENT{Note that, by definition, $ s\geq 4$; hence, $p \geq 2$.}
   \STATE Let $\mathcal{U}$ be the collection of all multisets of $\{1,2,\ldots, n\}$ of cardinality at most $\frac{ \kappa \ p }{\varepsilon^2}$, where $\kappa$ is a fixed constant. \label{step:unet}
   \STATE Write $C^i$ to denote the $i$th column of matrix $C=A+B$, for $i \in [n]$.
   \FORALL{ multisets $S \in \mathcal{U}$ } \label{step:loop}
   \STATE Set  $u = \frac{1}{|S|} \sum_{i \in S} C^i$. \\
   \COMMENT{$u$ is an $|S|$-uniform vector in the convex hull of the columns of $C$.}
    \STATE Solve convex program CP($u$).
   \IF{the objective function value of CP($u$) is \emph{less than} $\varepsilon/2$} \label{step:if}
   \STATE Return $(x,y)$, where $x$ and $y$ form an optimal solution of CP($u$).
   \ENDIF
   \ENDFOR
    \end{algorithmic}
\end{algorithm}

\begin{proof}[Proof of Theorem~\ref{thm:nash}]


Algorithm~\ref{alg:nash} iterates at most $n^{O\left( \frac{p}{\varepsilon^2}\right)}$ times, since this is an upper bound on the number of multisets of size $O(\frac{p}{\varepsilon^2})$. Furthermore, in each iteration the algorithm solves convex program CP($u$), this takes polynomial time. Given that $p=\log s$, these observations establish the desired running-time bound.

Now, in order to prove the theorem we need to show that Alogorithm~\ref{alg:nash} is (i) Sound: any mixed strategy pair, $(x,y)$, returned by the algorithm is an approximate Nash equilibrium; (ii) Complete: the algorithm always returns a mixed-strategy pair.

\emph{Soundness:} Lemma~\ref{lem:nash} implies that any mixed-strategy pair returned by the algorithm is guaranteed to be an $\varepsilon$-Nash equilibrium. Specifically, say for some $u$ the ``if'' condition in Step~\ref{step:if} is met. In addition, let $x$ and $y$ be the returned optimal solution of CP($u$). Then,
\begin{align}
| x^T ( Cy - u) | & \leq \| x \|_q \| Cy - u \|_p \qquad \textrm{(H\"{o}lder's inequality)} \\
& \leq 1 \times \varepsilon/2.
\end{align}
Here $q = p/(p-1) \geq 1$, since $p \geq 2$. The second inequality follows from the fact that the objective function value of CP($u$)  is no more than $\varepsilon/2$ and $\| x \|_q \leq \| x \|_1 = 1$. Since the returned $x$ satisfies the feasibility constraints in CP($u$) we have $x^T u \geq \pi_1 + \pi_2 - \varepsilon/2$.  Therefore, $x^T C y \geq x^T u - \varepsilon/2 \geq \pi_1 + \pi_2 - \varepsilon$. Overall, $x$ and $y$ satisfy the conditions in  Lemma~\ref{lem:nash}, and hence form an $\varepsilon$-Nash equilibrium.

\emph{Completeness:} It remains to show that the ``if'' condition in Step~\ref{step:if} is satisfied at least once (and hence the algorithm successfully returns a mixed strategy pair $(x,y)$). Next we accomplish this task.

Write $(\hat{x}, \hat{y})$ to denote a Nash equilibrium of the given game and let $\hat{\pi}_1$ ($\hat{\pi}_2$) be the payoff of the row (column) player under this equilibrium. Note that $C\hat{y}$ lies in the convex hull of the columns of $C$. Furthermore, since the sparsity of the game is $s$, for $p = \log s$, we have
\begin{align}
\| C^i \|_p & \leq 4 \qquad \forall i \in [n].
\end{align}

This follows from the fact that the entries of matrix $C$ lie between $-2$ and $2$ (recall that the payoffs are normalized between $-1$ and $1$); hence considering the $p$ norm of any column $i$ we get: $  \| C^i \|_p  \leq \left( 2^p s \right)^{1/p}  = 2 (s)^{1/p}  = 2 (2^{\log s})^{1/p}  = 4  $.

Therefore, for $p = \log s$, we can apply Theorem~\ref{thm:caratheodory-p} over the convex hull $\conv(\{C^i\}_i) $ with $\gamma \leq 4$. In particular, for $\mu = C \hat{y}$, Theorem~\ref{thm:caratheodory-p} implies that there exists a $O(\frac{p}{\varepsilon^2}) $ uniform vector $\mu'$ such that $\| C\hat{y} - \mu' \|_p \leq \varepsilon/2$

Since $\mu'$ is $O(\frac{p}{\varepsilon^2}) $ uniform, at some point during its execution the algorithm (with an appropriate value of $\kappa$) will set $u = \mu'$.  Therefore, at least once the algorithm will consider a $u$ that satisfies
\begin{align}
\| C\hat{y} - u \|_p & \leq \varepsilon/2. \label{ineq:goodu}
\end{align}

We show that in this case $\hat{x}$, $\hat{y}$, $\hat{\pi}_1$ and $\hat{\pi}_2 $ form a feasible solution of CP($u$) that achieves an objective function value of no more than $\varepsilon/2$. That is, the ``if'' condition in Step~\ref{step:if} is satisfied for this choice of $u$.

First of all the fact that the objective function value is no more than $\varepsilon/2$ follows directly from (\ref{ineq:goodu}).

Since $(\hat{x}, \hat{y})$ is a Nash equilibrium, using Theorem~\ref{thm:ms} we get
\begin{align}
\hat{x}^T C \hat{y} & = \hat{\pi}_1 + \hat{\pi}_2.
\end{align}

Next we show that $ \hat{x}^T u \geq \pi_1 + \pi_2 - \varepsilon/2$. Consider the following bound:
\begin{align}
| \hat{x}^T ( C\hat{y} - u ) | & \leq \| \hat{x} \|_q \| C\hat{y} - u \|_p \qquad \textrm{(H\"{o}lder's inequality)} \\
& \leq 1 \times \varepsilon/2.
\end{align}
Again, $q= p/(p-1) \geq 1 $ and we have $\| \hat{x} \|_q \leq 1$. Here, the second inequality now follows from our choice of $u$.  Since, $\hat{x}^T C \hat{y}  = \hat{\pi}_1 + \hat{\pi}_2$, we have $ \hat{x}^T u  \geq \hat{\pi}_1 + \hat{\pi}_2 - \varepsilon/2 $.

The remaining feasibility constraints of CP($u$) are satisfied as well. This simply follows from the fact that $\hat{x}, \hat{y}, \hat{\pi}_1$, and $\hat{\pi}_2$ form a feasible (in fact optimal) solution of (BP), see Theorem~\ref{thm:ms}.

Overall, we get that the ``if'' condition in Step~\ref{step:if} will be satisfied at least once and this completes the proof.



\end{proof}

\hfill \\

Ideas developed in this section can address other notions of sparsity as well. Specifically, if there exist $\alpha, \beta \in \mathbb{R}_+$ and $\gamma \in \mathbb{R}$ such that the matrix $\alpha A + \beta B + \gamma \mathbbm{1}_{n \times n}$ has column or row sparsity $s$, then our algorithm can be used to find an $\varepsilon$-Nash equilibrium of the game $(A,B)$ in time $n^{O\left(\frac{\lambda^2 \log s }{\varepsilon^2}\right)}$; here, $\mathbbm{1}_{n \times n}$ is the all-ones $n \times n$ matrix and $\lambda:= \max\{ \alpha, \beta, 1/\alpha, 1/\beta\}$. This follows from the fact that a $\left(\min\{ \alpha, \beta\} \ \varepsilon \right)$-Nash equilibrium of the game $\left(\alpha A, \beta B + \gamma \mathbbm{1}_{n \times n}\right)$ is an $\varepsilon$-Nash equilibrium of the game $(A,B)$. 

Furthermore, in time $n^{O\left(\frac{\log s}{\varepsilon^2}\right)}$, we can compute $\varepsilon$-Nash equilibria of games in which \emph{both} matrices $A$ and $B$ have column or row sparsity $s$. Note that this case is not a direct corollary of Theorem~\ref{thm:nash}. In particular,  if the columns of matrix $A$ and the rows of matrix $B$ are $s$ sparse, then it is not necessary that $A+B$ has low column or row sparsity. But, an approximate equilibrium of such a game can be computed by exhaustively searching for vectors, $v$ and $w$, that are $\varepsilon/4$-close (under the $\log s$-norm distance) to $A\hat{y}$ and $\hat{x}^TB$ respectively, here $(\hat{x}, \hat{y})$ is a Nash equilibrium of $(A,B)$.  In this case, instead of CP($u$), we need to solve a convex program that minimizes $\|Ay - v \|_{\log s} + \|B^T x - w \|_{\log s}$ and has the following constraint $x^Tv + w^T y \geq \pi_1 + \pi_2 - \varepsilon/2$ along with the constraints present in (BP). A direct extension of the proof of Theorem~\ref{thm:nash} shows that this program finds an $\varepsilon$-Nash equilibrium of the game $(A,B)$.

\begin{remark}
Consider the class of games in which the $p$ norm of the columns of matrix $C$ is a fixed constant. A simple modification of the arguments mentioned above shows that for such games an $\varepsilon$-Nash equilibrium can be computed in time $n^{O\left( \frac{p}{\varepsilon^2} \right)}$.
\end{remark}

\begin{remark}
Algorithm~\ref{alg:nash} can be adopted to find an approximate Nash equilibrium with large social welfare (the total payoffs of the players). Specifically, in order to determine whether there exists an approximate Nash equilibrium with social welfare more than $\alpha - \varepsilon$, we include the constraint $\pi_1 + \pi_2 \geq \alpha$ in CP($u$). The time complexity of the algorithm stays the same, and then via a binary search over $\alpha$ we can find an approximate Nash equilibrium with near-optimal social welfare.
\end{remark}

\begin{remark}
In Algorithm~\ref{alg:nash}, instead of the convex program CP($u$), we can solve the linear program with objective $\min \| Cy - u\|_\infty$ and constraints identical to CP($u$). Algorithm~\ref{alg:nash} still finds an approximate Nash, since $\| Cy - u \|_\infty \leq \|Cy - u  \|_p$ and we have $|x^T(Cy - u)| \leq \| x \|_1 \| Cy - u \|_\infty \leq 1 \times \varepsilon/2$ ($1$ and $\infty$ are H\"{o}lder conjugates of each other).

Solving a linear program, in place of a convex program,  would lead to a polynomial improvement in the running time of the algorithm. But, minimizing the $p$ norm of $Cy - u$ remains useful in specific cases; in particular, it provides a better running-time bound when the game is guaranteed to have a ``small probability'' equilibrium. We detail this result in the following Section.
\end{remark}

\subsection{Small Probability Games}
Daskalakis and Papadimitriou~\cite{smallprob} showed that there exists a PTAS for games that contain an equilibrium with small---specifically, $O\left(\frac{1}{n}\right)$---probability values. This result is somewhat surprisingly, since such small-probability equilibria have large---$\Omega(n)$---support, and hence are not amenable to, say, exhaustive search. This section shows that if a game has an equilibrium with probability values $O\left( \frac{1}{m} \right)$, for some $1 \leq m \leq n$, then an approximate equilibrium can be computed in time $n^{O(k/\varepsilon^2)}$, where $k$ has a logarithmic dependence on $s/m$. Since column sparsity $s$ is no more than $n$, we get back the result of~\cite{smallprob} as a special case.

\begin{definition}[Small Probability Equilibrium]
A Nash equilibrium $(x,y)$ is said to be $m$-small probability if all the entries of $x$ and $y$ are at most $\frac{1}{m}$.
\end{definition}

In~\cite{smallprob} a PTAS is given for games that have an equilibrium with probability values at most $\frac{1}{\delta n}$, for some fixed constant $\delta \in (0,1]$. Hence, the setting of~\cite{smallprob} corresponds to games that have $\delta n$-small probability equilibrium. Next we prove a result for general $m$-small probability equilibria.

\begin{theorem}
Let  $A, B \in [-1,1]^{n \times n}$ be the payoff matrices of an $s$-sparse bimatrix game. If $(A,B)$ contains an $m$-small probability Nash equilibrium, then an $\varepsilon$-Nash equilibrium of the game can be computed in time
\begin{align*}
n^{ O \left( \frac{t}{\varepsilon^2} \right) },
\end{align*}
where $t = \max \left\{ 2 \log\left( \frac{s}{m} \right) , 2 \right\}$.
\end{theorem}

\begin{proof}
Let norm $p= \max \left\{ 2 \log\left( \frac{s}{m} \right) , 2 \right\}$ and write $q$ to denote the H\"{o}lder conjugate of $p$, i.e., $q$ satisfies $\frac{1}{p} + \frac{1}{q} = 1$. To obtain this result Algorithm~\ref{alg:nash} is modified as follows: (i) use the updated value of $p$ instead of the one specified in Step~\ref{step:assign} of the algorithm; (ii) include convex constraint $\| x \|_q \leq {m^{-1/p}}$ in CP($u$); (iii) In Step~\ref{step:if} use $\frac{\varepsilon \ m^{1/p}}{2}$, instead of $\frac{\varepsilon}{2}$, as the threshold for returning a solution.

In order to establish the theorem we prove that the modified alogorithm is (i) Sound: any mixed strategy pair, $(x,y)$, returned by it is an approximate Nash equilibrium; (ii) Complete: the algorithm always returns a mixed-strategy strategy pair.

\emph{Soundness:} Below we show that if $x$ and $y$ are returned by this modified algorithm then they form a near-optimal solution of the bilinear program (BP) and, in particular, satisfy $x^T C y \geq \pi_1 + \pi_2 - \varepsilon $. Hence, Lemma~\ref{lem:nash} shows that any returned solution $(x,y)$ is an $\varepsilon$-Nash equilibrium.

Say the algorithm returns $(x,y)$ while considering vector $u$. Applying H\"{o}lder's inequality gives us:
\begin{align}
| x^T( Cy - u ) | & \leq \| x \|_q \| Cy - u \|_p \\
&  \leq m^{-1/p} \  \frac{\varepsilon \ m^{1/p}}{2} \\
& = \varepsilon/2.
\label{ineq:sound}
\end{align}

Here the second inequality uses the fact that $x$ satisfies the feasibility constraint $\|x \|_q \leq m^{-1/p}$ and the objective function value of CP($u$) is at most $\frac{\varepsilon \ m^{1/p}}{2}$. Since $x$ is a feasible solution of CP($u$), it satisfies $x^T u \geq \pi_1 + \pi_2 - \varepsilon/2$. Therefore, using inequality (\ref{ineq:sound}) we get $x^T C y \geq \pi_1 + \pi_2 - \varepsilon $, as required.

\emph{Completeness:} It remains to show that the ``if'' condition in Step~\ref{step:if} is satisfied at least once. To achieve this we prove that, for a particular $u$, an $m$-small probability Nash equilibrium forms a feasible solution of CP($u$) and achieves an objective function value of at most $\frac{\varepsilon \ m^{1/p}}{2}$. Therefore, for this $u$ the ``if'' condition in Step~\ref{step:if} will be met.




Recall that the columns of $C$ are $s$ sparse and its entries are at most $2$. Hence the $p$ norm of any column $C^i$ satisfies $\| C^i \|_p \leq (s 2^p)^{1/p} = 2 s^{1/p}$.

Let $(\hat{x}, \hat{y})$ be an $m$-small probability Nash equilibrium of the game. Theorem~\ref{thm:caratheodory-p}, applied over the convex hull $\conv(\{C^i\}_i) $ with $\gamma \leq 2 s^{1/p}$, guarantees that there exists a $O\left( \frac{ p \ s^{2/p} } {\varepsilon^2 \ m^{2/p}} \right)$ uniform vector that is $\frac{\varepsilon \ m^{1/p}}{2}$ close to $C\hat{y}$. Since $p \geq 2 \log \left( \frac{s}{m} \right)$, we have
\begin{align}
\frac{ s^{2/p} } { \ m^{2/p}} & = \left(\frac{s}{m}\right)^{2/p} \\
& \leq  2.
\end{align}

Therefore, there exists a $O\left(\frac{p}{\varepsilon^2}\right)$ uniform vector that is $\frac{\varepsilon \ m^{1/p}}{2}$ close to $C\hat{y}$. Such a vector, say $u$, will be selected by the algorithm in Step~\ref{step:loop} at some point of time. Below we show that, for $u$, the mixed strategies of the Nash equilibrium $\hat{x}$ and $\hat{y}$ are feasible solutions that achieve the desired objective function value.

To establish the feasibility of $\hat{x}$, we first upper bound its $q$ norm. The fact that its entries at most $1/m$ and $q = \frac{p}{p-1} \geq 1$ implies:

\begin{align}
\| \hat{x} \|_q & \leq \left( m \frac{1}{m^q} \right)^{1/q} \\
& = m^{-1/p}.
\end{align}

Since $(\hat{x},\hat{y})$ is a Nash equilibrium, there exists payoffs $\hat{\pi}_1$ and $\hat{\pi}_2$ such $\hat{x}^T C \hat{y} = \hat{\pi}_1 + \hat{\pi}_2$ (Theorem~\ref{thm:ms}). Also, $\hat{x}$, $\hat{y}$, $\hat{\pi}_1$, and $\hat{\pi}_2$ are feasible with respect to (BP). Hence, the only constraint of CP($u$) that we still need to verify is $\hat{x}^T u \geq \hat{\pi}_1 + \hat{\pi}_2 - \varepsilon/2$. This follows from H\"{o}lder's inequality:

\begin{align}
|\hat{x}^T(C\hat{y} - u ) | & \leq \| \hat{x} \|_q \ \| C\hat{y} - u \|_p \\
& \leq m^{-1/p} \ \frac{\varepsilon \ m^{1/p}}{2} \\
& = \varepsilon/2.
\end{align}

Overall, for the above specified $u$, the ``if'' condition in Step~\ref{step:if} will be satisfied. This shows that the algorithms successfully returns an $\varepsilon$-Nash equilibrium.

The algorithm iterates at most $n^{O\left( \frac{p}{\varepsilon^2}\right)}$ times, since this is an upper bound on the number of multisets of size $O(\frac{p}{\varepsilon^2})$. Furthermore, in each iteration the algorithm solves a convex program, this takes polynomial time. These observations establish the desired running-time bound and complete the proof.
\end{proof}

\section{Densest Subgraph}
\label{sect:dense}

This section and the next one present additive approximations for the normalized densest $k$-subgraph problem (\rm{NDkS}) and the densest $k$-bipartite subgraph problem (\rm{DkBS}) respectively.  

In \rm{NDkS} we are given a simple graph $G=(V,E)$ along with a size parameter $k \leq |V|$ and the goal is to find a maximum density subgraph containing exactly $k$ vertices. Here, density of a size-$k$ subgraph $S=(V_S, E_S)$ is defined to be $\rho(S):= |E_S|/k^2$. Note that in \rm{NDkS} density is normalized to be at most one. 

Next, we present a quadratic program, an $\varepsilon$-additive approximate solution of which can be used to efficiently find an $\varepsilon$-additive approximate solution of \rm{NDkS}. Write $A$ to denote the adjacency matrix of the given graph $G$ and let $n$ be the number of vertices in $G$. Define matrix $C : = \frac{1}{2} A + I$, here $I$ is the $n \times n$ identity matrix. 

\begin{align}
\max_{x} & \ \ \ x^T C x    \tag{\bf QP} \\
\textrm{subject to } \ \ \ & x_i \leq \frac{1}{k} \qquad \forall i \in [n]  \nonumber \\
& x \in \Delta^n \nonumber
\end{align}

Write $S^*$ to denote an optimal solution of the given \rm{NDkS} instance and let $z^*$ denote the optimal value of the quadratic program (QP). Given a solution $x$ of (QP) that achieves an objective function value of $z^* - \varepsilon$ (i.e., $x$ satisfies $x^T C x \geq z^* - \varepsilon$), we show how to efficiently find a subgraph $S$ that is an $\varepsilon$-additive approximate solution of \rm{NDkS}, i.e., find a size-$k$ subgraph $S$ that satisfies $\rho(S) \geq \rho(S^*) - \varepsilon$. Towards this end, the following lemma serves as a useful tool. 

\begin{lemma}
\label{lemma:unif}
Given a feasible solution, $y$, of the quadratic program (QP) we can find in polynomial time a feasible solution $z$ that satisfies  $z^T C z \geq y^T C y$ and, moreover, every  component of $z$ is either $0$ or $\frac{1}{k}$, i.e., $z_i \in \left\{0, \frac{1}{k} \right\}$ for each $i \in [n]$.
\end{lemma} 
\begin{proof}
Given feasible solution $y$ write $M(y) : = \{ i \in [n] \mid 0 < y_i < 1/k \}$. We iteratively update $y$ to decrease the cardinality of $M(y)$, and at the same time ensure that the objective function value (i.e., $y^T C y$) does not decrease. Since $|M(y)| \leq n$ we iterate at most $n$ times. This will overall establish the lemma. 

Note that each $i \in [n]$ indices both a component of $y$ and a vertex of the graph $G$. Write $\gamma_i$ to denote the total $y$ value of vertex $i$ and its neighbors in $G$, i.e., $\gamma_i := y_i + \sum_{j : (i,j) \in E } y_j$. Since $y \in \Delta^n$ and integer $k >1$, the cardinality of set $M(y)$ is either $0$ or strictly greater than one.  If $M(y) = \phi$ then the stated claims follows simply by setting $z=y$. Otherwise, if $|M(y)| \geq 2$ then we select two vertices $i ,j \in M(y)$ and change $y_i$ and $y_j$ such that the size of $M(y)$ decreases by at least one. We select $i, j \in M(y)$ as follows:

\begin{itemize}
\item If there exist vertices $i, j \in M(y)$ that are not connected by an edge in $G$, then without loss of generality we assume that $\gamma_i + y_i \geq \gamma_j + y_j $. 
\item If every pair of vertices $i, j \in M(y)$ is connected by an edge, then without loss of generality we assume that  $\gamma_i \geq \gamma_j$. 
\end{itemize}
 
Let $\delta : = \min \left\{ y_j , \frac{1}{k} - y_i \right\}$. We update $y_i \leftarrow y_i + \delta$ and $y_j \leftarrow y_j - \delta$, this update ensures that either $y_j$ goes down to zero or $y_i $ becomes equal to $1/k$; hence, the size of $M(y)$ decreases by at least one. 

Next we show this change in $y$ does not decrease the objective function value $y^T C y$. Thereby we get the stated claim. 

Consider the case in which $i$ and $j$ are not connected via an edge. The other case in which $(i,j) \in E$ (and we have $s_i \geq s_j$) follows along the same lines.   

Before the update the following inequality holds, $\gamma_i + y_i \geq \gamma_j + y_j $. For any $\delta$, the change in objective function value is equal to $(y_i + \delta) ( \gamma_i + \delta) + (y_j - \delta) (\gamma_j - \delta) - ( y_i \gamma_i  + y_j \gamma_j) $. This quality is equal to $\delta( y_i + \gamma_i - y_j - \gamma_j) + \delta^2$.   
Since $y_i + \gamma_i - y_j - \gamma_j \geq 0$ and $\delta$ is nonnegative, we get that the update in $y$ does not decrease the objective function value. This completes the proof. 
\end{proof}

Recall that $\rho(S^*)$ and $z^*$ denote the optimal values of the given \rm{NDkS} instance and (QP) respectively. Using Lemma~\ref{lemma:unif} we get the following proposition.
\begin{proposition}
\label{prop:qp}
The optimal value of (QP) is equal to the optimal value of the NDkS instance plus $1/k$, i.e., $z^* = \rho(S^*) + 1/k$.
\end{proposition}
\begin{proof}
We use $S^*$ to obtain a feasible solution, $\hat{x}$, for (QP) as follows: for each $i \in V(S^*)$ set $\hat{x}_i = 1/k$ and for every $ j \notin V(S^*)$ set $\hat{x}_j =0$. Here, $V(S^*)$ denotes the set of vertices of the subgraph $S^*$. 

By the definition of matrix $C$ we get that $\hat{x}^T C \hat{x} = \rho(S^*) + 1/k$. This implies that 
\begin{align}
z^* \geq \rho(S^*) + 1/k. \label{ineq:one}
\end{align} 

Using Lemma~\ref{lemma:unif} we can obtain an optimal solution $x'$ of the QP such that the components of $x'$ are either $0$ or $1/k$. By definition, $(x')^T C x'= z^*$. Write $S'$ to denote the subgraph induced by the vertex subset $\Supp(x')$. Note that $z^* =  (x')^T C x' = \rho(S') + 1/k$. Given that $S^*$ is an optimal solution of the \rm{NDkS} instance, we have $ \rho(S^*) \geq \rho(S')$. Hence, the following inequality holds 
\begin{align}
\rho(S^*) + 1/k \geq z^*. \label{ineq:two}
\end{align}

Inequalities (\ref{ineq:one}) and (\ref{ineq:two}) imply the stated claim.  

\end{proof}

Finally, we establish the connection between (QP) and \rm{NDkS}.
\begin{theorem}
\label{thm:qp}
Given an $\varepsilon$-additive approximate solution of (QP) we can find an $\varepsilon$-additive approximate solution of \rm{NDkS} in polynomial time. 
\end{theorem}
\begin{proof}
Given an $\varepsilon$-additive approximate solution of (QP), via Lemma~\ref{lemma:unif}, we can find an $\varepsilon$-additive approximation $\hat{x}$ whose components are either $0$ or $1/k$. Write $\hat{S}$ to denote the subgraph induced by vertex subset $\Supp(\hat{x})$; recall that, each  $i \in [n]$ indices both a component of $\hat{x}$ and a vertex of the graph $G$. 

Note that $ \hat{x}^T C \hat{x} = \rho(\hat{S}) + 1/k$. Since $\hat{x}^T C \hat{x} \geq z^* - \varepsilon$, the following inequality holds $\rho(\hat{S}) + 1/k \geq z^* - \varepsilon$. Using Propostion~\ref{prop:qp} we get that $\hat{S}$ is an $\varepsilon$-additive approximate solution of \rm{NDkS}, $\rho(\hat{S}) \geq \rho(S^*) - \varepsilon$. 
\end{proof}

In other words, in order to determine an approximate solution of \rm{NDkS} it suffices to compute an approximate solution of (QP). 




Note that, if the maximum degree of the graph $G$ is $d$ then the number of non-zero components in any column of the objective matrix $C$ is no more than $d+1$. This implies that, for $i \in [n]$, we have $\| C^i \|_p \leq (d +1)^{1/p}$. Here, $C^i$ denotes the $i$th column of $C$. Now for $p = \log (d +1) $ the following bound holds for all $i \in [n]$: $\| C^i \|_p \leq 2$. Therefore, as in Algorithm~\ref{alg:nash}, enumerating all $O\left( \frac{p}{\varepsilon^2}\right)$-uniform vectors in the convex hull of the columns of $C$, we can find an $\varepsilon$-approximate solution of (QP). This establishes the following theorem.

\begin{theorem}
Let $G$ be a graph with $n$ vertices and maximum degree $d$. Then, an $\varepsilon$-additive approximation of \rm{NDkS} over $G$ can be determined in time 
\begin{align*}
n^{O\left(\frac{\log d}{\varepsilon^2}\right)}.
\end{align*}
\end{theorem}

In general, this gives us an additive approximation algorithm for \rm{NDkS} that runs in quasi-polynomial time.

\subsection{Densest Bipartite Suubgraph}
In \rm{DkBS} we are given a graph $G=(V,E)$  along with a size parameter $k \leq |V|$ and the goal is to find size-$k$ vertex subsets, $S$ and $T$, such that the density of edges between $S$ and $T$ is maximized. Specifically, the (bipartite) density of vertex subsets $S$ and $T$ is defined as follows:
\begin{align}
\rho(S,T) & := \frac{|E(S,T)|}{|S| |T|},
\end{align}
here $E(S,T)$ denotes the set of edges that connect $S$ and $T$.

Let vertex subsets $S^*$ and $T^*$ form an optimal solution of the given \rm{DkBS} instance. Hence, $\rho(S^*, T^*)$ is equal to the optimal density. Next we state a bilinear program from~\cite{alonvempala} to approximate \rm{DkBS}. 
Here, $A$ denotes the adjacency matrix of the given graph $G$ and $n = |V|$.



\begin{align}
\max_{x,y} & \ \ \ x^TAy  \nonumber \\
\textrm{subject to } \ \ \ & x, y \in \Delta^n \nonumber \\
& x_i, y_i \leq \frac{1}{k} \qquad \forall i \in [n] \tag{BP-DkBS}.
\end{align}

Note that optimizing (BP-DkBS) over $x$, with a fixed $y$, corresponds to solving a linear program. Therefore, for any fixed $y$ there exists an optimal \emph{basic} feasible $x$, and vice versa. In other words, for any feasible pair $(x_0,y_0)$ we can find  $(x, y)$, such that $x_0^T A y_0 \leq x^TAy$ and the all the components of $x$ and $y$ are either $0$ or $1/k$. This observation implies that the optimal value of (BP-DkBS) is equal to $\rho(S^*, T^*)$. In addition, given an additive $\varepsilon$-approximate solution of (BP-DkBS), $(x',y')$, we can efficiently determine an $\varepsilon$-approximate solution of \rm{DkBS}. Specifically, we can assume without loss of generality that $x'$ and $y'$ are basic, and then for $S':=\Supp(x')$ and $T':=\Supp(y')$ we have $\rho(S',T') \geq \rho(S^*, T^*) - \varepsilon$. In other words, in order to determine an approximate solution of \rm{DkBS} it suffices to compute an approximate solution of (BP-DkBS). 

Note that the column sparsity of the objective matrix in (BP-DkBS)---i.e, the column sparsity of the adjacency matrix $A$---is equal to the maximum degree of the given graph. Therefore, as outlined in the previous section, we can modify Algorithm~\ref{alg:nash} to obtain the following theorem.





\begin{theorem}
Let $G$ be a graph with $n$ vertices and maximum degree $d$. Then, there exists an algorithm that runs in time $n^{O\left( \frac{\log d}{\varepsilon^2}\right)}$ and computes a $k \times k$-bipartite subgraph of density at least $\rho(S^*, T^*) - \varepsilon$.
\end{theorem}



\section{Extensions}\label{sect:ext}
This section extends Theorem~\ref{thm:caratheodory-p} to address convex hulls of matrices. Here we also detail approximate versions of  certain generalizations of the Carath\'{e}odory's theorem; specifically, we focus on the colorful Carath\'{e}odory theorem and Tverberg's theorem. 

First of all, note that a $d \times d$ matrix can be considered as a vector in $\mathbb{R}^{d^2}$, and hence directly applying Theorem~\ref{thm:caratheodory-p} to vectors in $\mathbb{R}^{d^2}$ we get results for entrywise $p$-norms. Recall that for a $d \times d$ matrix $Y$ the entrywise $p$-norm is defined as follows $\| Y \|_p := \left(\sum_{i=1}^d \sum_{j=1}^d |Y_{i,j}|^p \right)^{1/p}$.

In particular, we get the an approximate version of the Birkhoff-von Neumann theorem by considering $d \times d$ matrices as vector of size $d^2$ and applying Theorem~\ref{thm:caratheodory-p} with norm $p = \log d$. We know via the Birkhoff-von Neumann theorem that any $d \times d$ doubly stochastic matrix can be expressed as a convex combination of $d \times d$ permutation matrices (see, e.g.,~\cite{Barvinok}). The following corollary shows that for every doubly stochastic matrix $D$ there exists an $\varepsilon$-close (in the entrywise $\log d$-norm, and hence in the entrywise $\infty$-norm) doubly stochastic matrix $D'$ that can be expressed as a convex combination of $O\left( \frac{\log d}{\varepsilon^2} \right) $ permutation matrices. 

We say that a doubly stochastic matrix $D'$ is $k$ uniform if there exists a size $k$ multiset  $\Pi$ of permutation matrices such that $D'= \frac{1}{k} \sum_{P \in \Pi} P$.

\begin{corollary}
For every $d \times d$ doubly stochastic matrix $D$ there exists an $O\left( \frac{\log d}{\varepsilon^2} \right) $ uniform doubly stochastic matrix $D'$ such that $\max_{i,j} \ | D_{i,j} - D'_{i,j} | \leq \varepsilon$.
\end{corollary}

In addition to entrywise norms, we can establish an approximate version of Carath\'{e}odory's theorem for matrices under the Schatten $p$-norm.  Write $\| Y \|_{S_p}$ to denote that Schatten $p$-norm of a $d \times d'$ matrix $Y$, i.e., $ \| Y \|_{S_p} := \| \sigma(Y) \|_p$, where $\sigma(Y)$ is the vector of singular values of $Y$. Given a set of matrices $\mathcal{Y}=\{Y_1, Y_2,\ldots, Y_n \} \subset \mathbb{R}^{d \times d'}$, we say that a matrix $M' \in \conv(\mathcal{Y})$ is $k$ uniform (implicitly with respect to $\mathcal{Y}$) if there exists a size $k$ multiset $S$ of $[n]$ such that $M' = \frac{1}{k} \sum_{i \in S} Y_i$.

We can directly adopt the proof of Theorem~\ref{thm:caratheodory-p} and, in particular, use the matrix version of Khintchine inequality~\cite{So,tomczak} (instead of Theorem~\ref{thm:Khintchine-v}) to obtain the following result.

\begin{theorem}
\label{thm:caratheodory-matrix}
Given a set of matrices $\mathcal{Y}= \{ Y_1, Y_2,\ldots, Y_n \} \subset \mathbb{R}^{d \times d'} $ and $\varepsilon >0$. For every matrix $M \in \conv(\mathcal{Y})$ and $2 \leq p < \infty $ there exists an  $O\left(\frac{{p} \gamma^2}{\varepsilon^2}\right)$ uniform matrix $M' \in \conv(\mathcal{Y})$ such that $\| M - M' \|_{S_p} \leq \varepsilon$. Here, $\gamma := \max_{Y \in \mathcal{Y}} \  \| Y \|_{S_p}$.
\end{theorem}

\subsection{Approximating Colorful Carath\'{e}odory and Tverberg's Theorem}

The colorful Carath\'{e}odory theorem (see~\cite{Barany} and~\cite{Matouvsek}) asserts that if the convex hulls of $d+1$ sets (the color classes) $X_1, X_2, \ldots, X_{d+1} \subset \mathbb{R}^d$ intersect, then every vector in the intersection $\mu \in \cap_i \conv(X_i)$ can be expressed as a convex combination of vectors, each of which has a different color.   
\begin{theorem}[Colorful Carath\'{e}odory Theorem]
\label{thm:colorful-carat}
Let $X_1, X_2, \ldots, X_{d+1}$ be $d+1$ sets in $\mathbb{R}^d$ and vector $\mu \in \cap_i \conv(X_i)$. Then, there exists $d+1$ vectors $x_1, x_2, \ldots, x_{d+1}$ such that $x_i \in X_i$ for each $i$ and $\mu \in \conv(\{x_1, x_2, \ldots, x_{d+1} \})$.
\end{theorem}

For a given collection of $d+1$ sets (the color classes) $X_1, X_2, \ldots, X_{d+1} \subset \mathbb{R}^d$, a set $R \in \mathbb{R}^d$ is called a rainbow if $|R \cap X_i| = 1$ for all $i \in [d+1]$. 

The colorful Carath\'{e}odory  theorem guarantees the existence of a rainbow $R = \{x_1, x_2, \ldots, x_{d+1}\} $ whose convex hull contains the given vector $\mu$. But, determining a rainbow $R$ with $\mu \in \conv(R)$ in polynomial time remains an interesting open problem; see, e.g.,~\cite{Matouvsek}. Here we consider a natural approximate version of this question wherein the goal is to find a rainbow $R'$ whose convex hull is $\varepsilon$ close to $\mu$ in the following sense $\inf_{v \in \conv(R')} \| \mu - v \|_p  \leq \varepsilon$. Below we show that, for appropriately scaled vectors, Theorem~\ref{thm:caratheodory-p} can be used to efficiently determine such a rainbow $R'$. 

Theorem~\ref{thm:caratheodory-p} implies that there exists a vector $\mu'$ that satisfies $\| \mu - \mu'\|_p \leq \varepsilon $ and can be expressed as a convex combination of $t = O\left( \frac{p\gamma^2}{\varepsilon^2} \right)$ vectors of $R$, for $p \geq 2 $ and $ \gamma := \max_{x \in \cup_i X_i } \ \| x \|_p$. Let $R'$ be a size $t$ subset of $R$ such that $ \mu' \in \conv(R')$. Write $n = \sum_{i=1}^{d+1} |X_i|$ and note that we can exhaustively search for $R'$ in time ${d+1 \choose t} n^{O(t)} $. For a guessed $R'$ we can test if its convex hull is close to $\mu$ (i.e., $\inf_{v \in \conv(R')} \| \mu - v \|_p  \leq \varepsilon$) via a convex program. Moreover, after finding an appropriate subset $R'$ we can extend it to obtain a rainbow that is $\varepsilon$ close to $\mu$. 

\hfill \\

Another generalization of Carath\'{e}odory's theorem is Tverberg's Theorem (see~\cite{Tverberg} and~\cite{Matouvsek}), which is stated below. 

\begin{theorem}[Tverberg's Theorem]
Any set of $(r-1)(d+1) + 1$ vectors $X \subset \mathbb{R}^d $ can be partitioned into $r$ pairwise disjoint subsets $X_1, X_2, \ldots, X_r \subseteq X$ such that their convex hulls intersect: $\cap_{i=1}^r \conv(X_i) \neq \phi$. 
\end{theorem}

Note that here the special case of $r=2$ corresponds to Radon's theorem (see, e.g.,~\cite{Matouvsek}). 

Next we consider an approximate version of Tverberg's theorem. At a high level, our result guarantees the existence of pairwise disjoint subsets, $X_1',\ldots, X_r'$, that have small cardinality and whose convex hulls, $\conv(X_1'),\ldots, \conv(X_r')$, are ``concurrently close'' to each other. Formally:

\begin{definition}[Concurrently $\varepsilon$ close]
Sets $V_1, V_2, ..., V_r \subset \mathbb{R}^d$ are said to be concurrently $\varepsilon$ close under the $p$-norm distance if there exists a vector $\mu \in \mathbb{R}^d$ such that  $\inf_{v \in \conv(V_i)} \| \mu - v \|_p \leq \varepsilon$, for all $i \in [r]$.
\end{definition}

Tverberg's theorem together with Theorem~\ref{thm:caratheodory-p} establish the following result. 
\begin{theorem}
\label{thm:tverberg}
Let norm $p \in [2, \infty)$ and parameter $\varepsilon >0$. Then, any set of $(r-1)(d+1) + 1$ vectors $X \subset \mathbb{R}^d $ can be partitioned into $r$ pairwise disjoint subsets $X'_1, X'_2, \ldots, X'_r \subseteq X$ that are concurrently $\varepsilon$ close under the $p$-norm distance and satisfy $|X_i'| = O\left( \frac{p\gamma^2}{\varepsilon^2} \right)$, for all $i \in [r]$. Here, $\gamma:= \max_{x \in X} \| x \|_p$.
\end{theorem}
 
 
Tverberg's theorem guarantees the existence of a partition with subsets $X_1, X_2,\ldots, X_r$  whose convex hulls intersect. It is an interesting open problem whether such a partition can be determined in polynomial time; see~\cite{Matouvsek}. But, note that Theorem~\ref{thm:tverberg} implies that for fixed $r$ we can efficiently determine an approximate solution for this problem, i.e., determine a partition consisting of subsets that are concurrently $\varepsilon$ close. In particular, with $t = O\left( \frac{p\gamma^2}{\varepsilon^2} \right)$, we can exhaustively search for size $t$ subsets $X'_1, X'_2, \ldots, X'_r \subseteq X$ that are concurrently $\varepsilon$ close under the $p$-norm distance.\footnote{A polynomial-time algorithm exists for testing whether given subsets $X_1', X_2', \ldots, X_r'$ are concurrently $\varepsilon$ close under the $p$-norm distance. In particular, we can run the ellipsoid method with a separation oracle that solves the convex optimization problem $\min_{v \in \conv(X_i')} \| u - v \|_p$ to test if there exists a separating hyperplane for proposed ellipsoid center $u$.} This search runs in time $n^{O(rt)}$. Finally, we can extend the disjoint subsets $X'_1, X'_2, \ldots, X'_r \subseteq X$---by arbitrarily assigning the vectors in $X \setminus \left( \cup_i X_i' \right)$---to obtain a partition of $X$ that is concurrently $\varepsilon$ close.

\section{Lower Bound}
\label{sect:lb}
This section presents a lower bound that proves that in general $\varepsilon$-close (under the $p$-norm distance with $p \in [2, \infty)$) vectors cannot be expressed as a convex combination of less than $ \frac{1}{4 \varepsilon^{p/(p-1)}} $ vectors of the given set. This, in particular, implies that the $1/\varepsilon^2$ dependence established in Theorem~\ref{thm:caratheodory-p} is tight for the Euclidean case ($p=2$). 

To prove the lower bound we consider the convex hull of the standard basis vectors $e_1, e_2, \ldots, e_d \in \mathbb{R}^d$, i.e.,  consider the simplex $\Delta^d$. Write $\bar{e} \in \Delta^d$ to denote the uniform convex combination of $e_1, e_2,\ldots, e_d$: $\bar{e} := (1/d, 1/d, \ldots, 1/d)^T$. Next we show that any vector that can be expressed as a convex combination of less than $\frac{1}{4 \ \varepsilon^{p/(p-1)}} $ standard basis vectors is at least $\varepsilon$ away from $\bar{e}$ in the $p$-norm distance. This establishes the desired lower bound. 

Note that (trivially) any vector in $\Delta^d$ can be expressed as a convex combination of $\frac{1}{\varepsilon^{p/(p-1)}} $ standard basis vectors if, say, $ \frac{1}{\varepsilon^{p/(p-1)}} = d $. Hence, to obtain a meaningful lower bound we need to assume that $\frac{1}{\varepsilon^{p/(p-1)}}  $ is less than $d$.


\begin{proposition}
Let  $\frac{1}{\varepsilon^{p/(p-1)}}  < {d}$ and norm $p \geq 2$. Then, any vector $u \in \Delta^d$ that can be expressed as a convex combination of less than $\frac{1}{4 \ \varepsilon^{p/(p-1)}} $ standard basis vectors must satisfy $\| \bar{e} - u \|_p > \varepsilon $.
\end{proposition}
\begin{proof}
Write $k =\frac{1}{4 \ \varepsilon^{p/(p-1)}}  $. Since $\bar{e}$ is symmetric, we can assume without loss of generality that $u$ is a convex combination of vectors $e_1, e_2, \ldots, e_k$. The optimal value of the following convex program lower bounds the $p$-norm distance between $u$ and $\bar{e}$:

\begin{align*}
\min_{\alpha_1,\ldots,\alpha_k} & \ \ \ \left\| \bar{e} -  \sum_{i=1}^k \alpha_i e_i \right\|_p \nonumber \\
\textrm{subject to } \ \ \ &  \sum_{i=1}^k \alpha_i = 1 \\
& \alpha_i \geq 0 \qquad \forall i \in [k].
\end{align*}

K.K.T.~conditions imply that the optimal is achieved by setting $\alpha_i = 1/k$ for all $i \in [k]$. Hence, the $p$-norm distance between $\bar{e}$ and $u$ is at least $\left( k \left( \frac{1}{k} - \frac{1}{d} \right)^p \right)^{1/p}$. The choice of $k$ and the assumption that $\frac{1}{\varepsilon^{p/(p-1)}}  < {d}$ guarantee that $\| \bar{e} - u \|_p > \frac{3 \varepsilon}{ 4^{1/p} } $. Since $p \geq 2$, we get the desired claim.  


\end{proof}


\section*{Acknowledgements}
The author thanks Federico Echenique, Katrina Ligett, Assaf Naor, Aviad Rubinstein, Anthony Man-Cho So, and Joel Tropp for helpful discussions and references. This research was supported by NSF grants CNS-0846025 and CCF-1101470, along with a Linde/SISL postdoctoral fellowship.

\bibliographystyle{plain}
\bibliography{Caratheodory-Nash}

\end{document}